\documentclass[letterpaper, 10 pt, conference]{ieeeconf}  
\usepackage{graphicx}

\usepackage{amsmath,amsthm,amssymb,amsfonts}
\usepackage{xcolor}

\newcommand\rd{\mathrm{d}}
\newcommand\out{\mathrm{out}}

\newtheorem{theorem}{Theorem}[section]

\newtheorem{proposition}[theorem]{Proposition}
\newtheorem{assumption}{Assumption}

\title{\LARGE \bf
A KKL Observer Perspective on Reservoir Computing}

\author{ Anastasia Bizyaeva, Fernando Casta\~nos, Jaime A. Moreno
\thanks{A.B. is with the Sibley School of Mechanical and Aerospace Engineering, Cornell University, Ithaca, NY, USA, 14853 \texttt{anastasiab@cornell.edu}}%
\thanks{F.C. is with the Department of Automatic Control, Cinvestav, Mexico. \texttt{fernando.castanos@cinvestav.mx}}
\thanks{J.A.M. is with the Instituto de Ingenier\'ia, Universidad Nacional Aut\'onoma de M\'exico, Mexico. \texttt{JMorenoP@ii.unam.mx}}
}%

\begin{document}

\IEEEoverridecommandlockouts
\maketitle

\begin{abstract} Reservoir computing (RC) is a machine learning technique for data-driven modeling of dynamics for forecasting and control, primarily studied in computer science and physics literature with promising applications in neural network learning, physical computing, and neuroscience. Why reservoirs learn and how to choose good reservoir architectures are considered important open questions. We show that the RC problem is mathematically an extension of a classic problem in systems and control theory, the Kazantzis-Kravaris-Luenberger (KKL) observer design problem. As a consequence, many of the questions considered open for RC stand to benefit from a large body of theory in the mature KKL literature, non-exhaustively including on questions of embedding, transverse stability, local and global uniqueness guarantees, and effective data-driven solution constructions. Elaborating on this connection, we show that the surprising forecasting ability of reservoirs is in fact a direct consequence of the well-known observer internal model principle, derive an upper bound on the prediction error over a fixed forecast horizon, and provide a partial explanation for why linear readout training in RC works reasonably well. 
This work illustrates how classical ideas from systems and control can provide strong theoretical backing and open new questions for modern machine learning methods.
\end{abstract}
\section{Introduction}

In this paper we consider the \textit{reservoir computing (RC)} problem, a machine learning technique for data-driven modeling of time series that has been particularly successful at producing accurate long forecast of complex dynamics with minimal training \cite{jaeger2001echo,maass2002real,pathak2018model}. Primarily studied in physics and computer science literature, RC has promising applications in neural network training as well as physical computing and neuroscience, prompting many interdisciplinary investigations \cite{yan2024emerging,tanaka2019recent,enel2016reservoir}.
In RC, an observed time series drives a fixed high-dimensional dynamical system, called a \textit{reservoir}. 
Based on the reservoir's response to its driving signal, an output map is trained to reconstruct the signal and then fed back to the reservoir to generate forecasts autonomously beyond the observation horizon.

Despite substantial recent progress in establishing theoretical foundations for RC, including derived conditions for the echo state property, fading memory, universal approximation, and generalized synchronization \cite{jaeger2001echo,manjunath2013echo,grigoryeva2018echo,gonon2021fading,hart2020embedding,grigoryeva2023learning,hart2025generic}, the mechanisms underlying the success of reservoirs as \textit{dynamical models} are less well understood. In particular, existing theory provides limited insight into how reservoir synchronization with driving input and readout training combine to produce an autonomous system that reproduces and forecasts observed dynamics. Developing principled theoretical understanding and moving towards a constructive theory of RC where one can ask how to choose a reservoir that is well-suited to a particular dynamic task is considered a major open problem in machine learning for time series \cite{yan2024emerging}.

A closely related problem has long been studied in systems and control community through the theory of \textit{observers}, dynamical systems driven by measured outputs and designed to reconstruct unmeasured states of their driving system \cite{bernard2022}.  In particular, one popular observer design approach is the Kazantzis-Kravaris-Luenberger (KKL) observer which by construction embeds the plant dynamics into a stable latent system and then reconstructs the original system's state through an inverse map \cite{luenberger1966observers,kazantzis1998,andrieu2006}. KKL observer design bears a striking structural resemblance to the RC setup; in both problems, a measured signal drives a dynamical system (an observer or a reservoir) whose state converges to a representation of the underlying dynamics. Furthermore, both problems involve the construction of a readout map. In the classical KKL approach, the readout construction begins with a known plant model and the map is obtained by solution to a partial differential equation, from which the readout is then obtained as a left inverse. In contrast, RC constructs the readout using a model-free approach, by fitting the map to input-output data collected from the driven reservoir.

Our central insight in the presented work is that RC and KKL observer design are two sides of the same problem. Working in continuous time, we establish that contracting reservoirs are KKL output predictors whose readout maps are learned from data rather than constructed with the knowledge of the plant model. This equivalence explains reservoir forecasting through a connection to the observer internal model principle: by learning to reconstruct the observed signal, a reservoir acquires an internal model of the dynamics that generate it, thereby enabling prediction beyond the observation window. We then quantify how imperfect initialization and readout training affect forecast accuracy, and offer a partial explanation for why even simple linear readouts can perform well in RC applications. Together, these insights connect the empirical success of RC to the mature mathematical theory of observers and open a path towards a future constructive approach to reservoir computing leveraging classical insights from systems and control theory.

\subsection{Notation} We denote with $\mathbb{R}$ the real numbers. Define a square matrix $A \in \mathbb{R}^{n\times n}$. We denote by $\sigma(A)$ the set of all eigenvalues $\lambda_i$ of $A$. We say $A$ is Hurwitz if for any eigenvalue $\lambda_i \in \sigma(A)$, $\operatorname{Re}[\lambda_i] < 0$.  For a matrix $M \in \mathbb{R}^{n \times m}$ we denote by $M^{\dagger}$ the Moore-Penrose pseudoinverse. The singular values of $M\in\mathbb{R}^{a\times b}$ are the
$\min\{a,b\}$ nonnegative diagonal entries of $\Sigma$ in a
singular value decomposition $M=U\Sigma V^\top$.
We denote the smallest singular value by $\sigma_{\min}(M)$.We use $\|\cdot\|$ for the Euclidean norm of vectors and its induced matrix norm. For symmetric matrices $A$ and $B$ of the same size, $A \preceq B$ means that $B-A$ is positive semidefinite, while $A \prec B$ means that $B-A$ is positive definite. For $P=P^\top \succ 0$, the corresponding weighted norm is $\|v\|_P := \sqrt{v^\top P v}$.

For sets $\mathcal{X},\mathcal{Y}$, we use $C^1(\mathcal{X},\mathcal{Y})$ to denote the class of continuously differentiable functions from $\mathcal{X}$ to $\mathcal{Y}$.The notation $\emptyset$ denotes the empty set.
We denote by $\mathcal B(\mathbb R^d)$ the collection of Borel subsets of $\mathbb R^d$ defined as the smallest collection of open sets and closed under complements and countable unions. A measure on $\mathbb R^d$ is a function $\mu:\mathcal B(\mathbb R^d)\to[0,\infty)$ satisfying $\mu(\emptyset)=0$ and
$
\mu\left(\bigcup_{j=1}^{\infty}S_j\right)
=
\sum_{j=1}^{\infty}\mu(S_j)
$
for every sequence of pairwise disjoint Borel sets $S_j \in \mathcal B(\mathbb R^d)$. We call $\mu$ a probability distribution if $\mu(\mathbb R^d)=1$. For $z\in\mathbb R^d$ and $\varepsilon>0$, we denote by $B_\varepsilon(z)=\{w\in\mathbb R^d:\|w-z\|<\varepsilon\}$ the open Euclidean ball centered at \(z\). The support of $\mu$ is defined as 
\[
\operatorname{supp}\mu
=
\{z\in\mathbb R^d:\mu(B_\varepsilon(z))>0
\text{ for every }\varepsilon>0\}.
\]
Thus, a point belongs to the support if every open ball centered at that point has positive measure.

For a time-varying square integrable signal $v(t)$ over a time window $t \in [0,\tau]$ define  $\|v\|_\tau  := \left(\frac{1}{\tau}\int_0^\tau |v(t)\|^2\,dt\right)^{1/2}$.
A map $F$ is $L$-Lipschitz on a set $\mathcal{D}$ if
$\|F(v)-F(w)\| \leq L\|v-w\|$ for all $v,w \in \mathcal{D}$.
It is locally Lipschitz if this property holds in a
neighborhood of each point, with a possibly
neighborhood-dependent constant. For the autonomous system $\dot{x}=f(x)$, we denote by
$\phi_t(x_0)$ its unique solution at time $t$ initialized at
$x(0)=x_0$, whenever this solution exists.

\section{Observers and Reservoirs \label{sec:overview}}

First we will overview KKL observer design and reservoir computing, emphasizing their parallel constructions. For both problems, consider the autonomous system 
\begin{equation}
    \dot{x} = f(x), \quad y = h(x), \quad x(0) = x_0, \label{eq:sys}
\end{equation}
where $x \in \mathbb{R}^n$ is the state and $y \in \mathbb{R}^p$ is the output, with sufficiently smooth $f,h$. We will assume that the system \eqref{eq:sys} is well-behaved, i.e. that there exist sets $\mathcal{X}_0\subseteq \mathcal{X} \subseteq \mathbb{R}^n$ such that any solution $x(t)$ with $x(0) \in \mathcal{X}_0$ remains in $\mathcal{X}$ for all $t > 0$, where the set $\mathcal{X}$ is compact and hence the trajectories $x(t)$ are bounded. 

\subsection{Kazantzis-Kravaris-Luenberger (KKL) Observers}

An \textit{asymptotic observer} for \eqref{eq:sys} is a system
\begin{equation}
    \dot{z} = \phi(z,y), \quad \hat{x} = \psi(z), \quad z(0) = z_0 \label{eq:observer}
\end{equation}
where $z \in \mathbb{R}^m$ is the observer state, $\hat{x} \in \mathbb{R}^n$, with sufficiently smooth $\phi, \psi$, for which $\lim_{t\to\infty} \| x(t) - \hat{x}(t) \| = 0$ given any solution $x(t)$ of \eqref{eq:sys} in $\mathcal{X}$ \cite{bernard2022}. The challenge of observer design is obtaining the maps $\phi,\psi$ for which the error convergence between the reference system and the observer output is guaranteed. 

A popular observer construction, often referred to as the KKL observer, was first developed by Luenberger for linear systems \cite{luenberger1966observers} and then extended locally to nonlinear systems by Kazantzis and Kravaris \cite{kazantzis1998}, followed by global nonlinear extensions in \cite{kreisselmeier2003nonlinear,andrieu2006}. The KKL approach assumes $\phi$ in \eqref{eq:observer} to be a linear map and reframes the observer design problem as that of finding a positive integer $m$ and a nonlinear mapping $T: \mathbb{R}^n \to \mathbb{R}^m$ for which, if $x(t), y(t)$ are a solution of \eqref{eq:sys}, then $z = T(x)$ satisfies 
\begin{equation}
    \dot{z} = A z + B y, \quad \hat{x} = T^*(z) \label{eq:linear-observer}
\end{equation}
with $B \in \mathbb{R}^{m \times p}$, $T^*$ a sufficiently smooth left inverse of $T$, and $A \in \mathbb{R}^{m\times m}$ a Hurwitz matrix satisfying
\begin{equation}
    A^{\top} P + P A \preceq - \alpha I \label{eq:Hurwitz_cond}
\end{equation}
for some symmetric positive definite $P$ and positive real $\alpha$. A key insight of KKL theory is that such a map $T$ can be constructed as a solution to a partial differential equation. Suppose we can find a function $T \in C^1(\mathbb{R}^n,\mathbb{R}^m)$ such that
\begin{equation}
 \frac{ \partial T(x)}{ \partial x} f(x) = AT(x) + B h(x) . \label{eq:KKL_PDE}
\end{equation}
The error $\varepsilon =  z - T(x)$ evolves according to
\begin{align*}
 \dot{\varepsilon} &= Az + B h(x) - \frac{\partial T(x)}{\partial x} f(x) \\
         &= Az + B h(x) - \left( AT(x) + B h(x) \right) = A \varepsilon .
\end{align*}
Since $A$ is Hurwitz, the error converges to zero exponentially. If, further, $T$ is injective, then $T^\ast$ can be chosen as a left inverse of $T$, that is, $T^\ast(T(x)) = x$, and the convergence of the error to zero implies the convergence of $\hat{x}$ to $x$. We will refer to a KKL observer constructed for a linear reference plant as a Luenberger observer, as this was the original construction in \cite{luenberger1966observers}. Constructing a solution to the PDE \eqref{eq:KKL_PDE} can be challenging, and is an active area of investigation. Some numerical strategies can be found in \cite{ramos2020numerical}.

More recently, motivated in part by connections between observer design and recurrent neural networks in machine learning, 
\cite{pachy2024existence,lucia2026learning} introduced a KKL observer construction with nonlinear contracting dynamics. In this generalization, the observer dynamics \eqref{eq:observer} satisfy for all $(z,y)$ the condition 
\begin{equation}
    P \frac{\partial \phi}{\partial z}(z,y) + \frac{\partial \phi}{\partial z}(z,y)^{\top} P \preceq - \alpha I \label{eq:contract_cond}
\end{equation}
for some symmetric positive definite $P$ and positive real $\alpha$. Then a solution $T(x)$ to the PDE
\begin{equation}
\frac{\partial T}{\partial x}(x) f(x) = \phi(T(x), h(x))\label{eq:KKL_PDE_contracting}
\end{equation}
which admits a sufficiently smooth left inverse $T^*(z)$ defines the observer as $\hat{x} = T^*(z)$, since the error $\varepsilon = z - T(x)$ evolves according to 
\begin{align*}
    \dot{\varepsilon} &=  \phi(z,h(x)) - \frac{\partial T(x)}{\partial x} f(x) \\
         &= \phi(z,h(x)) - \phi(T(x),h(x)).
\end{align*}
Letting $P$ be as in \eqref{eq:contract_cond} we can choose a Lyapunov function $V(\varepsilon) = \varepsilon^T P \varepsilon$, the derivative of which along the error dynamics can be shown to satisfy $\dot{V} \leq - \alpha \| \varepsilon \|^2 \leq - \frac{\alpha}{\lambda_{max}(P)}V$ which guarantees asymptotic (in fact, exponential) convergence of the state estimate $\hat{x}$ to the true state $x$ as $t \to \infty$.

\subsection{Reservoir Computing}
A \emph{reservoir} \cite{jaeger2001echo,maass2002real} is an input-driven dynamical system\footnote{Although reservoir computing is commonly formulated in discrete time due to the popularity of the approach for recurrent neural network training, here we adopt a continuous-time formulation as in \cite{hart2024generalised} to establish a direct parallel with continuous time KKL observer constructions. Analogous insights can be developed in discrete time.}
\begin{equation} 
\dot{r} = g(r, u), \quad r(0) = r_0 \label{eq:reservoir}
\end{equation} 
where $r \in \mathbb{R}^{n_r}$ is the reservoir state, $u \in \mathbb{R}^{n_u}$ is the reservoir input, and $g$ is a sufficiently smooth and expressive map. Typically $n_r > n_u$, and we assume $g(r,u)$ to be Lipschitz in $u$. In typical constructions such as the Echo State Network of \cite{jaeger2001echo}, $g$ is realized by a recurrent neural network with fixed, randomly initialized weights. Reservoir computing (RC) is a machine learning technique in which we receive observations of a dynamic input $u(t)$ and the corresponding forced reservoir response history $r(t)$ from an initial condition $r(0) = r_0$ over a time interval $t \in [0,\tau]$. The goal of RC is to construct a data-driven model of the process generating the input time history $u(t)$ from this input-output relationship, typically for the purposes of  forecasting or control. This is achieved by finding an output feedback law $\omega_{\rm out}: \mathbb{R}^{n_r} \to \mathbb{R}^{n_u}$ that augments the reservoir dynamics as 
\begin{equation}
\dot{r} = g(r,\hat{u}), \quad \hat{u} = \omega_{\rm out}(r) \label{eq:reservoir_autonomous}
\end{equation}
where $\hat{u} \in \mathbb{R}^{n_u}$ acts as a model for the time series $u(t)$ beyond the training horizon $\tau$.

A key challenge in RC is identifying properties of the reservoir and the input time series that enable the learning of the output map $ \omega_{\rm out}$ in  \eqref{eq:reservoir_autonomous}. For temporal forecasting, a useful framework for analyzing this problem is \textit{generalized synchronization (GS)} between the reservoir and the source dynamics \eqref{eq:sys} \cite{kocarev1996generalized}. We say the driven reservoir exhibits GS for input $u(t) = y(t) = h(x(t))$ with $p = n_u$ if there exists a continuous map $H:\mathcal{X} \to \mathbb{R}^{n_r}$ such that $r = H(x)$ is invariant along the trajectories and $\lim_{t \to \infty} \|r(t) - H(x(t))\| = 0$ for initial conditions $r_0$ in some neighborhood. The GS property for a reservoir therefore characterizes its ``forgetting'' of initial conditions $r_0$ within a synchronization basin of attraction. It has been shown in \cite{hart2024generalised} that the GS map is defined through the PDE condition
\begin{equation}
    \frac{\partial H}{\partial x}(x) f(x) = g(H(x),h(x)). \label{eq:GS_PDE}
\end{equation}
The \textit{echo state property (ESP)}  \cite{yildiz2012re,manjunath2013echo} imposes a corresponding uniqueness requirement on the reservoir response. A reservoir has ESP if for each admissible bounded input history $u(t)$ on $(-\infty,t]$ there must be a unique bounded reservoir response $r(t)$ on that interval. In other words, the current reservoir state must be uniquely determined by its input history. A sufficient condition for a reservoir to have both GS and ESP is uniform contraction of reservoir dynamics \cite{hart2024generalised}, i.e. the reservoir equations satisfy 
\begin{equation}
    P \frac{\partial g}{\partial r}(r,u) + \frac{\partial g}{\partial r}(r,u)^{\top} P \preceq - \alpha I  \label{eq:contraction_reservoir}
\end{equation}
for some symmetric positive definite $P$ and positive real $\alpha$ for all $r \in \mathbb{R}^{n_r}$, $u \in \mathbb{R}^{n_u}$. This contraction condition \eqref{eq:contraction_reservoir} is often used a design criterion for reservoir initialization.

Consider forecasting the output of the autonomous system \eqref{eq:sys}. In the RC framework, for $t \in [0,\tau]$, observations of the original dynamics force the reservoir as inputs, with $u(t) = y(t)$ in \eqref{eq:reservoir} and the reservoir response $r(t)$ is collected. Then for $t > \tau$, the trained autonomous  reservoir \eqref{eq:reservoir_autonomous} is simulated forward and its outputs $\hat{u}(t)$ serve as predictions of observations $y(t)$ of \eqref{eq:sys}. On the GS graph, an exact ouptut readout satisfies $\omega_{\rm out}(H(x)) = h(x)$. If $H$ is injective, this read out is $h \circ H^{-1}$ on $H(\mathcal{X})$. 
The problem of finding the reservoir feedback map $\omega_{\rm out}$ can be stated as a minimization over an appropriate function space 
\begin{equation}
\omega_{\rm out} =  \operatorname*{argmin}_{\omega} \int_0^{\tau} \| u(t) -\omega( r(t)) \|^2 \rd t\label{eq:omega_out} 
\end{equation}
as in \cite{ohkubo2024reservoir}. In RC literature, it is further predominantly assumed that the output map $\omega_{\rm out}$ can be sufficiently well approximated as a linear map $\omega_{\rm out}(r) = W_{\rm out} r$ with $W_{\rm out} \in \mathbb{R}^{n_u \times n_r}$. Then the feedback design problem simplifies to a least squares problem 
\begin{equation}
W_{\out} = \operatorname*{argmin}_{W \in \mathbb{R}^{n_u \times n_r}} \int_0^{\tau} \| u(t) - W r(t) \|^2 \rd t\label{eq:lsq_prob} 
\end{equation} 
whose minimum-Frobenius-norm solution is
\begin{equation} 
 W_{\out} = \left( \int_0^{\tau} u(t) r(t)^T \rd t \right)\left( \int_0^{\tau} r(t) r(t)^T \rd t \right)^{\dagger}.\label{eq:Wout} 
\end{equation}
A reservoir trained with a least squares approximation then produces forecasts according to 
\begin{equation}
\dot{r} = g(r,\hat{u}), \quad \hat{u} = W_{\rm out} r. \label{eq:reservoir_autonomous_linearmap}
\end{equation}
\subsection{A brief history of convergence between KKL and RC}
An interesting convergent evolution has been happening between the RC and KKL observer literatures, with the two fields rarely cross-referencing one another. 
From the inception of RC in \cite{jaeger2001echo,maass2002real}, reservoir dynamics $g$ have been assumed nonlinear, even when the output map is restricted to be linear. This is in contrast with classical KKL observers \cite{kazantzis1998,andrieu2006}, where observer dynamics are linear while the output map is nonlinear. A nonlinear observer dynamics construction was only recently considered in the KKL context with the contracting observers of \cite{pachy2024existence,lucia2026learning}, while a nonlinear output map construction was only recently considered in the RC context, e.g. in \cite{ohkubo2024reservoir,grigoryeva2023learning,grigoryeva2024data}. A number of works have also recently appeared on model-free KKL observer constructions, e.g. \cite{ramos2020numerical,miao2023learning,tang2023data,tang2024synthesis,peralez2024deep}. Meanwhile in RC literature, classically the problem is posed with full state observation $y = x$ in \eqref{eq:sys}, but was recently generalized to ``reservoir observers'' for reconstruction of unseen variables, further blurring the lines between the two fields \cite{lu2017reservoir,liu2026reservoir}.  Many mathematical results were also discovered in parallel, for example an embedding principle for discrete-time KKL observers was obtained in \cite{brivadis2019luenberger} and an analogous result was stated for reservoirs four years later in \cite{grigoryeva2023learning}.

In a most striking parallel, recently the exact reservoir computing problem was independently derived in the KKL literature under a different name.
In \cite{janny2021deep} the authors design an \textit{output predictor} using the input/output response data of a uniformly contracting dynamical system driven by the output of an autonomous plant \eqref{eq:sys}. In the proposed framework, a feedback law is learned which, in closed loop with the original contracting system, synthesizes a generative model that forecasts the output $y$ of \eqref{eq:sys} for times beyond the training horizon. This framework rediscovers from a starting point grounded in observer design theory the \textit{exact} setup exploited by reservoir computing since its inception in \cite{jaeger2001echo,maass2002real}. The authors in \cite{janny2021deep} further suggest using the linear KKL framework with a nonlinear readout as particular model choice for the output predictor design \cite{jaeger2001echo}. However, \cite{janny2021deep} stops at one-dimensional output prediction and does not formulate the multidimensional case. In the remainder of this paper we will make explicit the connection between the RC and KKL frameworks, reinterpreting RC as a KKL output predictor as coined in \cite{janny2021deep}. Our goal here is to join together insights from these two previously dijoint bodies of scholarship and build upon them.


\section{Reservoirs are KKL Output Predictors \label{sec:equivalence}}

From the setup in Section \ref{sec:overview}, we immediately notice several strong mathematical parallels between KKL observers and reservoirs. The embedding map $T$ in KKL observer design is precisely the generalized synchronization map $H$ in reservoirs. The training step in RC is a data-driven construction of an observer over the window $[0,\tau]$. Contraction conditions \eqref{eq:linear-observer},\eqref{eq:contract_cond},\eqref{eq:contraction_reservoir} ensure, for both the observer and the reservoir, that the embedding/generalized synchronization map exists and is unique on a relevant neighborhood of interest. In this section we make these parallels concrete and illustrate that indeed, KKL obsevers and reservoirs are fundamentally interconnected constructions: in the infinite data limit on the synchronized manifold, the output map $\omega_{\rm out}$ learned a by a reservoir can be constructed from the solution to the KKL PDE \eqref{eq:KKL_PDE},\eqref{eq:KKL_PDE_contracting} and, conversely, a solution to the KKL PDE induces a closed-loop feedback law generates forward predictions of the system output $y$ in \eqref{eq:sys}. 

 To establish our main result in this section we will first introduce a few additional notations and assumptions. To make more precise the reservoir training loss in \eqref{eq:omega_out}, let $\mu$ be a probability distribution on $( \mathbb{R}^{n_r} \times \mathbb{R}^{n_u}$ describing the ampling of the measured reservoir state and input pairs $(r,u)$ and define the function class $\mathcal{W}$ to consist of locally Lipschitz functions $\mathbb{R}^{n_r} \to \mathbb{R}^{n_u}$. The loss function that defines the general reservoir training objective is 
\begin{equation}
\mathcal{L}_{\mu}(\omega) = \int \|u - \omega(r) \|^2 \rd\mu, \quad \omega_{\rm out}  = \operatorname*{argmin}_{\omega \in \mathcal{W}} \mathcal{L}_{\mu}(\omega).  \label{eq:reservoir_loss}
\end{equation}

The following assumptions clarifies the conditions under which this learning objective is well-posed.
\begin{assumption}[Well-Posedness] \label{ass:learning_well_posedness}
Given generalized synchronization map $H(x)$ defined through a solution to \eqref{eq:GS_PDE} and the plant output map $h(x)$ from \eqref{eq:sys}, the following hold. 

A. For the probability measure $\mu$ in \eqref{eq:reservoir_loss}, we have $\operatorname{supp}(\mu) = \{ (H(x),h(x)): x \in \mathcal{X} \}$;

B.  There exists $\omega^* \in \mathcal{W}$ satisfying $\omega^*(H(x)) = h(x)$ for all $x \in \mathcal{X}$.

\end{assumption}
\noindent Assumption \ref{ass:learning_well_posedness} is purposefully general, leaving our setup agnostic to the specific properties of the plant dynamics \eqref{eq:sys}.  The first condition in the assumption is a data informativity requirement, as it requires every neighborhood of an input/output pair $(r,u)$ of the driven reservoir on the generalized synchronization manifold associated with the compact set $\mathcal{X}$ of \eqref{eq:sys} to receive positive training weight. This means, in turn, that zero training loss directly implies exact output reconstruction throughout the entire synchronization set. The second part of the assumption ensures realizability within the specified function class $\mathcal{W}$. This will hold, for eample, if the generalized synchronization map $H$ is injective with a Lipschitz inverse on $H(\mathcal{X})$ and if $h$ is Lipschitz on $\mathcal{X}$.  
With these two general conditions, $h\circ H^{-1}$ admits a Lipschitz extension to the reservoir state space $\mathbb{R}^{n_r}$. In the linear case $H(x)=\Pi x$ and $h(x)=Gx$, full column rank of $\Pi$ gives the explicit readout $\omega^*(r)=G\Pi^\dagger r$. Interestingly, injectivity of $H$ is sufficient but not necessary for the assumption, and therefore the output reconstruction, to hold as was also observed in \cite{janny2021deep}. For example if $H(x) = (h(x), \ell(x))$  then a properly constructed coordinate projection would be sufficient for ouput reconstruction.

Finally, before stating our main result we introduce a convenient notion of distance between two readouts $\omega_1,\omega_2 \in \mathcal{W}$ on the reservoir's synchronization manifold, given a trajectory $x(t)$ of the system \eqref{eq:sys} over time window $t \in [0,\tau]$:
\begin{equation}
d_{\tau} (\omega_1, \omega_2) = \left( \frac{1}{\tau} \int_0^{\tau} \| \omega_1(H(x(t))) - \omega_2(H(x(t))) \|^2 \rd t \right)^{1/2} \label{eq:output_dist}
\end{equation}
We are now ready to establish our general equivalence result.

\begin{proposition}[Contracting reservoir is a contracting KKL output predictor] \label{thm:equivalence} Consider the interconnection of \eqref{eq:sys} and \eqref{eq:reservoir} with $u = y$, with  \eqref{eq:sys} on a compact forward invariant set $\mathcal{X}$,  with the reservoir equations \eqref{eq:reservoir} satisfying the uniform contraction condition \eqref{eq:contraction_reservoir}, and under Assumption \ref{ass:learning_well_posedness}. We will assume all functions to be Lipschitz when appropriate.

\noindent \textbf{A.}  For the driven reservoir from initial conditions $x(0),r(0)$,
\begin{equation*}
\|r(t) - H(x(t))\|_P \leq e^{- \lambda t} \| r(0) - H(x(0))\|_P,
\end{equation*}
 where $\lambda = \alpha/(2 \lambda_{max}(P))$ with $\alpha,P$ defined through the contraction condition \eqref{eq:contraction_reservoir};

\noindent\textbf{B.} Given a synchronized initialization $r(0) = H(x(0))$ and for any $\omega \in \mathcal{W}$, the following are equivalent: 
\begin{enumerate}
\item $\omega$ is a global minimizer of \eqref{eq:reservoir_loss}; 
\item $\omega(H(x))  = h(x)$ for every $x \in \mathcal{X}$; 
\item For every $x_0 \in \mathcal{X}$, the closed-loop reservoir
\begin{equation*}
\dot{r} = g(r,\omega(r)), \quad \hat{u} = \omega(r), \quad r(0) = H(x_0)
\end{equation*}
satisfies $r(t) = H(\phi_t(x_0))$ and $\hat{u}(t) = h(\phi_t(x_0))$ for all $t \geq 0$ where $\phi_t$ is the flow of \eqref{eq:sys} on $ \mathcal{X}$ from initial condition $x_0$.
\end{enumerate}
\textbf{C.} Consider a finite time horizon $0 < \tau < \infty$ and let $\mathcal{W}_{L} \subset \mathcal{W}$ be a class of locally Lipschitz functions $\omega: \mathbb{R}^{n_r} \to \mathbb{R}^{n_u}$ with a common Lipschitz constant $L$ on a neighborhood of $H(\mathcal{X})$. Let $\omega^*$ defined in Assumption 1.B be in $\mathcal{W}_L$ and consider training data $r(t),u(t) = h(x(t))$ generated over the time window $t \in [0,\tau]$.\ from initial condition satisfying $\| r(0) - H(x(0))\| < \varepsilon$. Then the minimizer 
\begin{equation}
\hat{\omega}_{\varepsilon,\tau} \in \operatorname*{argmin}_{\omega \in \mathcal{W}_L} \frac{1}{\tau} \int_0^{\tau} \| h(x(t)) - \omega(r(t)) \|^2 \rd t \label{eq:approximate_output}
\end{equation}
when it exists satisfies 
\begin{equation}
d_{\tau}(\hat{\omega}_{\varepsilon,\tau},\omega^*) \leq 2 L \kappa  
\varepsilon \sqrt{\frac{1 - e^{-2 \lambda \tau}}{2 \lambda \tau}}
\end{equation}
where $\lambda = \alpha/(2 \lambda_{max}(P))$ is the reservoir contraction rate and $\kappa = \sqrt{\lambda_{max}(P)/\lambda_{min}(P)}$. Hence,given any $\delta > 0$ the learned readout satisfies $d_{\tau}(\hat{\omega}_{\varepsilon,  \tau}) \leq \delta$ whenever $\varepsilon \leq (\delta/2 L \kappa) \sqrt{2 \lambda \tau/(1 - e^{-2 \lambda \tau})}$.
\end{proposition}
\begin{proof}
A. This is a restatement of a result independently established  in the reservoir computing literature \cite{hart2024generalised} and in KKL observer literature \cite{andrieu2006}.
B. By Assumption 1, $\mathcal L_\mu(\omega^*)=0$, so the minimum loss is zero. Let $\omega$ be a global minimizer and suppose, for contradiction,  there exists $x^*\in\mathcal X$ such that
\begin{equation}
\varepsilon:=\|h(x^*)-\omega(H(x^*))\|>0.
\end{equation}
By continuity, there exists $\delta>0$ such that $\|u-\omega(r)\|\ge\varepsilon/2$ for every $(r,u)\in B_\delta(r^*)$, where $r^*=(H(x^*),h(x^*))$. Assumption 1.A gives $r^*\in\operatorname{supp}\mu$, hence $\mu(B_\delta(r^*))>0$. Therefore,
\begin{equation}
\mathcal L_\mu(\omega)
\ge
\int_{B_\delta(r^*)}\|u-\omega(r)\|^2\,d\mu
\ge
\frac{\varepsilon^2}{4}\mu(B_\delta(r^*))
>0,
\end{equation}
contradicting minimality. This proves $1)\Rightarrow2)$.

Conversely, suppose $\omega(H(x))=h(x)$ for every $x\in\mathcal X$. By Assumption 1.A, every pair $(r,u)\in\operatorname{supp}\mu$ has the form $(H(x),h(x))$, so $u-\omega(r)=0$ on the support. Since $\mu$ assigns zero measure outside its support, $\mathcal L_\mu(\omega)=0$. Nonnegativity of the loss then implies that $\omega$ is a global minimizer, proving $2)\Rightarrow1)$.

Suppose $2)$ holds. For any $x_0\in\mathcal X$, define $\bar r(t)=H(\phi_t(x_0))$. The generalized synchronization PDE \eqref{eq:GS_PDE} gives
\begin{equation}
\dot{\bar r}(t)
=g(\bar r(t),h(\phi_t(x_0)))
=g(\bar r(t),\omega(\bar r(t))).
\end{equation}
Thus $\bar r(t)$ solves the closed-loop reservoir equations with the prescribed initial condition. Local Lipschitz continuity of the closed-loop vector field ensures uniqueness, yielding $r(t)=\bar r(t)$ and $\hat u(t)=h(\phi_t(x_0))$ for all $t\ge0$. This proves $2)\Rightarrow3)$. 

Finally, evaluating the output identity in $3)$ at $t=0$ gives
\begin{equation}
\omega(H(x_0)) =\omega(r(0)) =\hat u(0) =h(\phi_0(x_0)) =h(x_0).
\end{equation}
Since $x_0\in\mathcal X$ is arbitrary, this proves $3)\Rightarrow2)$.

C. Let $\bar r(t)=H(x(t))$ as before. By part A and equivalence of the Euclidean and $P$-weighted
norms,
\begin{equation*}
    \|r(t)-\bar r(t)\|
    \leq \kappa e^{-\lambda t}\|r(0)-\bar r(0)\|
    \leq \kappa\varepsilon e^{-\lambda t}.
\end{equation*}
Consequently,
\begin{equation*}
    \|r-\bar r\|_\tau^2 \leq \frac{\kappa^2 \varepsilon^2}{\tau} \int_0^\tau e^{-2 \lambda t} \rd t = \kappa^2\varepsilon^2
   \frac{1-e^{-2\lambda\tau}}{2\lambda\tau}.
\end{equation*}
For sufficiently small $\varepsilon$, both $r(t)$ and $\bar r(t)$ lie in the neighborhood on which the readouts have common Lipschitz constant $L$. Since $ \omega^*\in\mathcal W_L$ and $\omega^*(\bar r(t))=h(x(t))=u(t)$, optimality of
$\hat\omega_{\varepsilon,\tau}$ implies
\begin{multline*}
    \|u-\hat\omega_{\varepsilon,\tau}(r)\|_\tau
    \leq \|u-\omega^*(r)\|_\tau\ \\
= \|\omega^*(\bar r)-\omega^*(r)\|_\tau
    \leq L\|\bar r-r\|_\tau.
\end{multline*}
Using the triangle inequality and the Lipschitz bound for
$\hat\omega_{\varepsilon,\tau}$, we obtain
\begin{multline*}
    d_\tau(\hat\omega_{\varepsilon,\tau},\omega^*) = \|\hat\omega_{\varepsilon,\tau}(\bar r)-u\|_\tau \\
\leq
    \|\hat\omega_{\varepsilon,\tau}(\bar r)
      -\hat\omega_{\varepsilon,\tau}(r)\|_\tau+
    \|\hat\omega_{\varepsilon,\tau}(r)-u\|_\tau\\
    \leq 2L\|\bar r-r\|_\tau
    \leq 2L\kappa\varepsilon
    \sqrt{\frac{1-e^{-2\lambda\tau}}{2\lambda\tau}}.
\end{multline*}
The stated condition ensuring
$d_\tau(\hat\omega_{\varepsilon,\tau},\omega^*)\leq\delta$
follows by rearranging the terms in the statement above.
\end{proof}

Proposition \ref{thm:equivalence} carries the following intuitive interpretation. First, Proposition \ref{thm:equivalence}. A shows that during training, a reservoir's internal dynamics reach synchronization with its driving input and in fact, if equipped with the learned readout $\omega_{\rm out}$ found in Proposition \ref{thm:equivalence} 
This synchronization phenomenon is well-known in both the KKL and reservoir computing literature, e.g. \cite{jaeger2001echo,hart2023generalised}. 

Proposition  \ref{thm:equivalence}.B establishes that reservoirs trained on outputs from an autonomous system learn an internal model of the output dynamics in order to generate closed-loop forecasts. This observation is related to the famous internal model principle in systems and control \cite{francis1976internal,trumpf2014internal,bin2022internal} which, broadly speaking, states that a dynamical system must contain within it of the dynamics it seeks to asymptotically reproduce (as is the case for observers) or regulate (as is the case for controllers). For KKL observers, the internal model is established by construction; we showed that for reservoirs an internal model emerges from training the readout map on observed output trajectories, and then realized in forward simulation by closed-loop reservoir dynamics. The classical idea of internal model principle therefore provides a principled systems-theoretic explanation for the forecasting capability of reservoirs, a question that is considered a central puzzle in reservoir computing \cite{bollt2021explaining,yan2024emerging}.

Finally, Proposition \ref{thm:equivalence}.C  establishes that reservoirs trained with input/output data that is not perfectly synchronized still learn an output map that is reasonably close to the map obtained under idealized conditions. In the following section we will build on this and establish an upper bound for the error of a forecast of a trajectory of \eqref{eq:sys} obtained via a closed loop reservoir with an imperfect learned output map and with imperfect (i.e. not perfectly synchronized) initialization.

\section{Forecast error upper bound}

In this section we consider forecasts produced by reservoirs whose ouput readout map was fit to (possibly imperfect) data. We will establish an error bound, generalizing results in \cite{janny2021deep} which itself synthesizes the developments in \cite{marconi2007output,isidori2010existence}. In \cite{janny2021deep} under assumption of scalar output, a forecasting error for contracting output predictors is established under perfect knowledge of the GS map, and bounds accounting imperfect output map were established for KKL output predictors with linear dynamics and nonlinear output map. We will generalize the proof of \cite{janny2021deep} to establish an error bound in the most general case relevant to RC, for multidimensional inputs and imperfect learned readout map with potentially nonlinear, contracting reservoir dynamics. Let $\omega_\theta \in \mathcal{W}$ be an approximate readout map obtained by solving the minimization problem \eqref{eq:reservoir_loss} with associated Lipschitz constant $L_{\omega}$, and let $L_u$ be the Lipschitz constant for the open loop reservoir $g(r,u)$ with respect to $u$, and let $L_{g}$ be the Lipschitz constant for the closed loop reservoir map $g(r,\omega_{\rm out}(r))$.

\begin{proposition}[Forecast Error Bound] \label{prop:forecast_bound}Consider a reservoir trained over the time window $t \in [0,\tau]$ to obtain a readout map $\hat{u}_{\theta}(t) = \omega_{\theta}(r(t))$.
Suppose $\|\omega_{\theta}( H(x)) - h(x) \| \leq \varepsilon$ for $x \in \mathcal{X}$ and nonnegative $\varepsilon$. Then
\begin{multline}
\|\hat{u}_{\theta}(\tau + s) - h(x(\tau+s))\|  \leq  \varepsilon \left( 1 + \frac{L_{\omega} L_u}{L_{g}} \left( e^{L_{g} s} - 1 \right) \right) \\
+  \kappa L_{\omega} e^{-\lambda \tau + L_{g} s} \|r(0) - H(x(0))\|  \label{eq:forecast_bound}
\end{multline}
where $\lambda = \alpha/(2 \lambda_{max}(P))$ is the reservoir contraction rate and $\kappa = \sqrt{\lambda_{max}(P)/\lambda_{min}(P)}$.
\end{proposition}

\begin{proof}
Let $\bar r(t)=H(x(t))$, and let $r_\theta(t)$ denote the forecast trajectory satisfying $\dot r_\theta=g(r_\theta,\omega_\theta(r_\theta)), r_\theta(\tau)=r(\tau).$ Define $e(t)=r_\theta(t)-\bar r(t)$ for $t\geq\tau$. By Proposition III.1.A and equivalence of norms,
\begin{equation*}
    \|e(\tau)\|
    \leq \kappa e^{-\lambda\tau}
    \|r(0)-H(x(0))\|.
\end{equation*}
Since $\dot{\bar r}=g(\bar r,h(x))$, the Lipschitz bounds and the assumed readout error give
\begin{align*}
    \|\dot e(t)\|
    &\leq
    \|g(r_\theta,\omega_\theta(r_\theta))
      -g(\bar r,\omega_\theta(\bar r))\|\\
    &\quad+
    \|g(\bar r,\omega_\theta(\bar r))
      -g(\bar r,h(x))\| \leq L_g\|e(t)\|+L_u\varepsilon.
\end{align*}
Integrating and applying Gronwall's inequality yields
\begin{equation*}
    \|e(\tau+s)\|
    \leq e^{L_gs}\|e(\tau)\|
    +\frac{L_u\varepsilon}{L_g}
    \left(e^{L_gs}-1\right).
\end{equation*}

Finally, since $\hat u_\theta=\omega_\theta(r_\theta)$, the Lipschitz property of $\omega_\theta$ and the assumed readout error bound imply
\begin{equation*}
    \|\hat u_\theta(\tau+s)-h(x(\tau+s))\| \quad\leq L_\omega\|e(\tau+s)\|+\varepsilon
\end{equation*}
from which we obtain the stated bound \eqref{eq:forecast_bound}.
\end{proof}

The upper bound \eqref{eq:forecast_bound} on the forecasting error is made up of two terms: the first term captures the effect of mismatch of initial conditions from the synchronization manifold, and the second term captures the effects of imperfectly learned feedback map. Both are a source of error, and the upper bound on this error that we obtained grows exponentially in time.
Notice that synchronized initialization requires that $r(0) = H(x(0))$, and perfect synchronization map learning means $\varepsilon = 0$. In this idealized limit, both terms on the right hand side of the bound inequality \eqref{eq:forecast_bound} are zero, and we recover the perfect indefinite forecasting on the synchronization manifold that was established in Proposition \ref{thm:equivalence}.B.

\section{Why linear regression works for RC training}

The most pronounced difference between the KKL and RC approaches is how the output map $\omega_{\rm out}$ is constructed in practical applications. In the KKL literature, both model-based and model-free constructions focus on approximating the solution to the nonlinear PDE \eqref{eq:KKL_PDE}, which can be challenging in general. Meanwhile, a puzzling insight from the RC literature is that frequently, one can get away with simply learning a linear feedback law via least squares training \eqref{eq:lsq_prob},\eqref{eq:Wout}. This is despite the fact that RC typically deploys nonlinear reservoirs whose synchronization map must satisfy an even more complcated nonlinear PDE condition \eqref{eq:GS_PDE}. Moreover, the efficient least squares training step is often advertised as a selling point of RC over more traditional gradient-based approaches for training RNNs. Reservoirs trained this way have been shown empirically to work robustly once appropriately tuned, including to produce long forecasts of complex spatiotemporally chaotic dynamics \cite{pathak2018model}. This empirical evidence seemingly contradicts intuition from Proposition \ref{prop:forecast_bound} where we saw that mismatch between the true feedback map on the GS manifold and an approximate one can propagate forecasting error. To reconcile these two observations, we next ask, how close to linearity are the solutions to the contracting KKL/GS PDE \eqref{eq:GS_PDE} and the readout map $\omega_{\rm out}$ constructed from them? 

To attempt to address this question head on we first introduce some additional notation. Suppose $H(x)$ is the generalized synchronization map for \eqref{eq:sys},\eqref{eq:reservoir}, and suppose we want to approximate it with a linear map $\bar{H}(x) = \Pi x$ with $\Pi \in \mathbb{R}^{n_r \times n}$. Substituting this approximation into the GS PDE \eqref{eq:GS_PDE} we can write down the PDE residual 
\begin{equation}
    \rho_{\Pi}(x) = \Pi f(x) - g(\Pi x, h(x)), \quad x \in \mathcal{X}. \label{eq:residual}
\end{equation}
Then we can define the magintude of the residual \eqref{eq:residual} as 
\begin{equation}
    R(\Pi) = \sup_{x \in \mathcal{X}} \| \rho_{\Pi}(x) \|. \label{eq:residual_norm}
\end{equation}
Additionally, for a matrix $W \in \mathbb{R}^{n_u \times n_r}$, define
\begin{equation}
    \eta(\Pi,W)
    = \sup_{x \in \mathcal{X}} \|h(x)-W\Pi x\|.
    \label{eq:eta}
\end{equation}
The following proposition bounds the departure of $H$ from linearity and the error of the linear readout learned from data.

\begin{proposition}[Gap from Linearity of KKL PDE Solutions and Output Readout] Consider \eqref{eq:sys},\eqref{eq:reservoir} under the hypotheses of Proposition \ref{thm:equivalence}. Assume additionally that $\mathcal{X}$ is invariant for all $t \in \mathbb{R}$ for the dynamics \eqref{eq:sys} and that $H$ is a $C^1$ solution of \eqref{eq:GS_PDE}. Then, for every $\Pi \in \mathbb{R}^{n_r \times n}$,

\begin{equation}
    \sup_{x \in \mathcal{X}} \|H(x) - \Pi x\| \leq \frac{\kappa R(\Pi)}{\lambda} \label{eq:GS_linearity_bound}
\end{equation}
where $R(\Pi)$ is defined by \eqref{eq:residual_norm}, $\lambda = \alpha/(2 \lambda_{max}(P))$ is the reservoir contraction rate and $\kappa = \sqrt{\lambda_{max}(P)/\lambda_{min}(P)}$.
Furthermore, let $\omega_{\out}$ be a global minimizer of \eqref{eq:reservoir_loss}, and let $W_{\out}$ solve \eqref{eq:lsq_prob} using data collected over $[0,\tau]$, with $\tau>0$ and synchronized initialization $r(0)=H(x(0))$. The learned linear readout
$\omega_\theta(r)=W_{\out}r$ satisfies
\begin{equation}
    d_\tau(\omega_\theta,\omega_{\out})
    \leq \inf_W \Bigl\{\eta(\Pi,W)
    +\frac{\kappa R(\Pi)}{\lambda}\|W\|\Bigr\},
\label{eq:trained_readout_bound}
\end{equation}
where the infimum is over
$W \in \mathbb{R}^{n_u \times n_r}$ and $d_\tau$ is defined
in \eqref{eq:output_dist}.
\end{proposition}

\begin{proof}
   Define $E=\sup_{x \in \mathcal{X}}\|H(x)-\Pi x\|$, which is finite by compactness and continuity. Along any complete plant trajectory in $\mathcal{X}$,  by \eqref{eq:GS_PDE} and \eqref{eq:residual} , $H(x(t))$ and $\Pi x(t)$ solve
\begin{equation*}
    \dot r=g(r,h(x(t)))+v(t)
\end{equation*}
with respective additive inputs $v(t)=0$ and $v(t)=\rho_\Pi(x(t))$. Applying \cite[Theorem~37(ii)]{davydov2022non} in the $P$-weighted norm over $[-s,0]$, using norm equivalence, and taking the supremum over
$x(0)\in\mathcal{X}$ gives
\begin{equation*}
    E \leq \kappa e^{-\lambda s}E
    +\frac{\kappa R(\Pi)}{\lambda}
    (1-e^{-\lambda s}), \qquad s>0.
\end{equation*}
Letting $s\to\infty$ we then arrive at \eqref{eq:GS_linearity_bound}.
By Proposition~\ref{thm:equivalence}.A--B, synchronized initialization gives $r(t)=H(x(t))$ and $\omega_{\out}(r(t))=u(t)$. Thus, for any comparison matrix $W$, optimality in \eqref{eq:lsq_prob} and the triangle inequality yield
\begin{multline*}
    d_\tau(\omega_\theta,\omega_{\out}) =\|u-W_{\out}r\|_\tau \leq \|u-Wr\|_\tau \\
    \leq \|h(x)-W\Pi x\|_\tau +\|W\|\,\|H(x)-\Pi x\|_\tau  \\
    \leq \eta(\Pi,W)  +\frac{\kappa R(\Pi)}{\lambda}\|W\|.
\end{multline*}
Taking the infimum over $W$ we arrive at \eqref{eq:trained_readout_bound}.
\end{proof}

The bound \eqref{eq:trained_readout_bound} clarifies how ``far'' a trained linear readout is from a true nonlinear reconstruction map. We saw in Proposition \ref{prop:forecast_bound} that a large mismatch propagates to a large error bound. We want this quantity to then be as small as possible to enable effective forecasting with linear least squares training. The term $\eta(\Pi,W)$ in the bound measures the mismatch of the output map $h(x)$ in the linear reconstruction, while $\|W\| \kappa R(\Pi)/\lambda$ bounds the amplification of the encoding error due to the mismatch in geometry between the synchronization map and the learned readout. From the bound statement we see that stronger contraction rate alone does not necessarily guarantee better forecasting accuracy, since all other quantities in the bound also evolve with reservoir parameters. This upper bound provides a potential constructive condition that can be used to tune reservoir parameters to promote accurate forecasting with linear least squares training, for example using optimization. It also provides a partial explanation for the empirical success of least-squares trained reservoirs, since an approximately linear synchronization map, together with accurate output reconstruction using weights of moderate norm, guarantees a small error for the learned linear readout.

\section{Discussion}

In this paper we derived an interpretation of reservoir computing, a machine learning technique, through the lens of KKL observer theory. We hope this work contributes to an increased dialogue between practitioners of RC and the observer design community within systems and control. Although these fields have both evolved largely in parallel, their closely related mathematical foundations offer substantial opportunities for future collaborative progress. The mature KKL literature offers decades of results on embedding, reconstruction, convergence, and numerical construction of observers that can inform both the analysis and the design of reservoirs. 
Conversely, empirical evidence of efficient readout training and the successful usage of diverse reservoir architectures suggest new directions for observer theory and implementation. The RC connection also points to how observer theory could guide the selection and tuning of physical reservoirs, harnessing their intrinsic dynamics for signal reconstruction and prediction, expanding the impact of observer theory to the domain of nonconventional computing. In neuroscience, through the RC connection observer theory can offer a framework for investigating how reservoir-like neural circuits could build internal representations of their environment from sensory input, with learned readouts and feedback supporting estimation and prediction. In this paper we make explicit the shared strucutre of the RC and KKL problems, hoping to facilitate the transfer of existing insights and stimulate new questions on both sides.





\section*{Acknowledgments}
During the preparation of this paper, the authors became aware of closely related work developed independently and concurrently  \cite{2026Saqib}. We thank the authors for a friendly correspondence upon discovering the overlap, and we look forward to future collaboration.

The authors acknowledge the use of ChatGPT Astra for helping identify several references that proved useful in developing the ideas in the manuscript, and for feedback on correctness of several mathematical statements. The authors take full responsibility for the originality, accuracy, and integrity of the presented work.

\bibliographystyle{ieeetr}
\bibliography{references}

@article{pathak2018model,
  title={Model-free prediction of large spatiotemporally chaotic systems from data: A reservoir computing approach},
  author={Pathak, Jaideep and Hunt, Brian and Girvan, Michelle and Lu, Zhixin and Ott, Edward},
  journal={Physical review letters},
  volume={120},
  number={2},
  pages={024102},
  year={2018},
  publisher={APS}
}

@article{jaeger2001echo,
  title={The ``echo state” approach to analysing and training recurrent neural networks-with an erratum note},
  author={Jaeger, Herbert},
  journal={Bonn, Germany: German National Research Center for Information Technology GMD Technical Report},
  volume={148},
  number={34},
  pages={13},
  year={2001},
  publisher={Bonn}
}

@article{bollt2021explaining,
  title={On explaining the surprising success of reservoir computing forecaster of chaos? The universal machine learning dynamical system with contrast to {VAR} and {DMD}},
  author={Bollt, Erik},
  journal={Chaos: An Interdisciplinary Journal of Nonlinear Science},
  volume={31},
  number={1},
  year={2021},
  publisher={AIP Publishing}
}

@article{hart2023generalised,
  title={Generalised synchronisations, embeddings, and approximations for continuous time reservoir computers},
  author={Hart, Allen G},
  journal={Physica D: Nonlinear Phenomena},
  pages={133956},
  year={2023},
  publisher={Elsevier}
}

@article{grigoryeva2023learning,
  title={Learning strange attractors with reservoir systems},
  author={Grigoryeva, Lyudmila and Hart, Allen and Ortega, Juan-Pablo},
  journal={Nonlinearity},
  volume={36},
  number={9},
  pages={4674},
  year={2023},
  publisher={IOP Publishing}
}

@article{grigoryeva2018echo,
  title={Echo state networks are universal},
  author={Grigoryeva, Lyudmila and Ortega, Juan-Pablo},
  journal={Neural Networks},
  volume={108},
  pages={495--508},
  year={2018},
  publisher={Elsevier}
}

@article{gonon2021fading,
  title={Fading memory echo state networks are universal},
  author={Gonon, Lukas and Ortega, Juan-Pablo},
  journal={Neural Networks},
  volume={138},
  pages={10--13},
  year={2021},
  publisher={Elsevier}
}

@article{hart2020embedding,
  title={Embedding and approximation theorems for echo state networks},
  author={Hart, Allen and Hook, James and Dawes, Jonathan},
  journal={Neural Networks},
  volume={128},
  pages={234--247},
  year={2020},
  publisher={Elsevier}
}

@article{davydov2022non,
  title={Non-Euclidean contraction theory for robust nonlinear stability},
  author={Davydov, Alexander and Jafarpour, Saber and Bullo, Francesco},
  journal={IEEE Transactions on Automatic Control},
  volume={67},
  number={12},
  pages={6667--6681},
  year={2022},
  publisher={IEEE}
}

@article{enel2016reservoir,
  title={Reservoir computing properties of neural dynamics in prefrontal cortex},
  author={Enel, Pierre and Procyk, Emmanuel and Quilodran, Ren{\'e} and Dominey, Peter Ford},
  journal={PLoS computational biology},
  volume={12},
  number={6},
  pages={e1004967},
  year={2016},
  publisher={Public Library of Science}
}

@article{hart2024generalised,
  title={Generalised synchronisations, embeddings, and approximations for continuous time reservoir computers},
  author={Hart, Allen G},
  journal={Physica D: Nonlinear Phenomena},
  volume={458},
  pages={133956},
  year={2024},
  publisher={Elsevier}
}

@article{kocarev1996generalized,
  title={Generalized synchronization, predictability, and equivalence of unidirectionally coupled dynamical systems},
  author={Kocarev, Ljupco and Parlitz, Ulrich},
  journal={Physical review letters},
  volume={76},
  number={11},
  pages={1816},
  year={1996},
  publisher={APS}
}

@article{yildiz2012re,
  title={Re-visiting the echo state property},
  author={Yildiz, Izzet B and Jaeger, Herbert and Kiebel, Stefan J},
  journal={Neural networks},
  volume={35},
  pages={1--9},
  year={2012},
  publisher={Elsevier}
}

@article{grigoryeva2024data,
  title={Data-driven cold starting of good reservoirs},
  author={Grigoryeva, Lyudmila and Hamzi, Boumediene and Kemeth, Felix P and Kevrekidis, Yannis and Manjunath, G and Ortega, Juan-Pablo and Steynberg, Matthys J},
  journal={arXiv preprint arXiv:2403.10325},
  year={2024}
}

@article{tanaka2019recent,
  title={Recent advances in physical reservoir computing: A review},
  author={Tanaka, Gouhei and Yamane, Toshiyuki and H{\'e}roux, Jean Benoit and Nakane, Ryosho and Kanazawa, Naoki and Takeda, Seiji and Numata, Hidetoshi and Nakano, Daiju and Hirose, Akira},
  journal={Neural Networks},
  volume={115},
  pages={100--123},
  year={2019},
  publisher={Elsevier}
}

@article{maass2002real,
  title={Real-time computing without stable states: A new framework for neural computation based on perturbations},
  author={Maass, Wolfgang and Natschl{\"a}ger, Thomas and Markram, Henry},
  journal={Neural computation},
  volume={14},
  number={11},
  pages={2531--2560},
  year={2002},
  publisher={MIT Press}
}

@article{yan2024emerging,
  title={Emerging opportunities and challenges for the future of reservoir computing},
  author={Yan, Min and Huang, Can and Bienstman, Peter and Tino, Peter and Lin, Wei and Sun, Jie},
  journal={Nature Communications},
  volume={15},
  number={1},
  pages={2056},
  year={2024},
  publisher={Nature Publishing Group UK London}
}

@article{manjunath2013echo,
  title={Echo state property linked to an input: Exploring a fundamental characteristic of recurrent neural networks},
  author={Manjunath, Gandhi and Jaeger, Herbert},
  journal={Neural computation},
  volume={25},
  number={3},
  pages={671--696},
  year={2013},
  publisher={MIT Press}
}

@article{kazantzis1998,
title = {Nonlinear observer design using {Lyapunov's} auxiliary theorem},
journal = {Systems \& Control Letters},
volume = {34},
number = {5},
pages = {241-247},
year = {1998},
issn = {0167-6911},
doi = {https://doi.org/10.1016/S0167-6911(98)00017-6},
url = {https://www.sciencedirect.com/science/article/pii/S0167691198000176},
author = {Nikolaos Kazantzis and Costas Kravaris}
}

@article{andrieu2006,
author = {Andrieu, Vincent and Praly, Laurent},
title = {On the Existence of a {Kazantzis}--{Kravaris}/{Luenberger} Observer},
journal = {SIAM Journal on Control and Optimization},
volume = {45},
number = {2},
pages = {432-456},
year = {2006},
doi = {10.1137/040617066}
}

@article{bernard2022,
title = {Observer design for continuous-time dynamical systems},
journal = {Annual Reviews in Control},
volume = {53},
pages = {224-248},
year = {2022},
issn = {1367-5788},
doi = {https://doi.org/10.1016/j.arcontrol.2021.11.002},
author = {Pauline Bernard and Vincent Andrieu and Daniele Astolfi}
}

@article{pachy2024existence,
  title={On the existence of {KKL} observers with nonlinear contracting dynamics},
  author={Pachy, Victor and Andrieu, Vincent and Bernard, Pauline and Brivadis, Lucas and Praly, Laurent},
  journal={IFAC-PapersOnLine},
  volume={58},
  number={21},
  pages={262--267},
  year={2024},
  publisher={Elsevier}
}

@article{luenberger1966observers,
  title={Observers for multivariable systems},
  author={Luenberger, David},
  journal={IEEE transactions on automatic control},
  volume={11},
  number={2},
  pages={190--197},
  year={1966},
  publisher={IEEE}
}

@article{lucia2026learning,
  title={Learning a Contracting {KKL}-observer with Local Optimal Guarantees},
  author={Luc{\'\i}a Galimberti, Clara and Peralez, Johan and Astolfi, Daniele and Andrieu, Vincent and Nadri, Madiha},
  journal={arXiv},
  year={2026}
}

@article{trumpf2014internal,
  title={Internal model principles for observers},
  author={Trumpf, Jochen and Trentelman, Harry L and Willems, Jan C},
  journal={IEEE Transactions on Automatic Control},
  volume={59},
  number={7},
  pages={1737--1749},
  year={2014},
  publisher={IEEE}
}

@article{francis1976internal,
  title={The internal model principle of control theory},
  author={Francis, Bruce A and Wonham, Walter Murray},
  journal={Automatica},
  volume={12},
  number={5},
  pages={457--465},
  year={1976},
  publisher={Elsevier}
}

@article{bin2022internal,
  title={Internal models in control, bioengineering, and neuroscience},
  author={Bin, Michelangelo and Huang, Jie and Isidori, Alberto and Marconi, Lorenzo and Mischiati, Matteo and Sontag, Eduardo},
  journal={Annual Review of Control, Robotics, and Autonomous Systems},
  volume={5},
  number={1},
  pages={55--79},
  year={2022},
  publisher={Annual Reviews}
}

@article{kreisselmeier2003nonlinear,
  title={Nonlinear observers for autonomous Lipschitz continuous systems},
  author={Kreisselmeier, Gerhard and Engel, Robert},
  journal={IEEE Transactions on Automatic Control},
  volume={48},
  number={3},
  pages={451--464},
  year={2003},
  publisher={IEEE}
}

@article{ohkubo2024reservoir,
  title={Reservoir computing with generalized readout based on generalized synchronization},
  author={Ohkubo, Akane and Inubushi, Masanobu},
  journal={Scientific reports},
  volume={14},
  number={1},
  pages={30918},
  year={2024},
  publisher={Nature Publishing Group UK London}
}

@inproceedings{brivadis2019luenberger,
  title={Luenberger observers for discrete-time nonlinear systems},
  author={Brivadis, Lucas and Andrieu, Vincent and Serres, Ulysse},
  booktitle={Conference on Decision and Control},
  pages={3435--3440},
  year={2019},
  organization={IEEE}
}

@article{lu2017reservoir,
  title={Reservoir observers: Model-free inference of unmeasured variables in chaotic systems},
  author={Lu, Zhixin and Pathak, Jaideep and Hunt, Brian and Girvan, Michelle and Brockett, Roger and Ott, Edward},
  journal={Chaos: An Interdisciplinary Journal of Nonlinear Science},
  volume={27},
  number={4},
  year={2017},
  publisher={AIP Publishing}
}

@article{liu2026reservoir,
  title={Reservoir observer enhanced with residual calibration and attention mechanism},
  author={Liu, Yichen and Xiao, Wei and Chu, Tianguang},
  journal={Physical Review E},
  volume={113},
  number={3},
  pages={034201},
  year={2026},
  publisher={APS}
}

@inproceedings{peralez2024deep,
  title={Deep model-free {KKL} observer: A switching approach},
  author={Peralez, Johan and Nadri, Madiha},
  booktitle={6th Annual Learning for Dynamics \& Control Conference},
  pages={929--940},
  year={2024},
  organization={PMLR}
}

@inproceedings{miao2023learning,
  title={Learning robust state observers using neural odes},
  author={Miao, Keyan and Gatsis, Konstantinos},
  booktitle={Learning for Dynamics and Control Conference},
  pages={208--219},
  year={2023},
  organization={PMLR}
}

@article{tang2023data,
  title={Data-driven state observation for nonlinear systems based on online learning},
  author={Tang, Wentao},
  journal={AIChE Journal},
  volume={69},
  number={12},
  pages={e18224},
  year={2023},
  publisher={Wiley Online Library}
}

@inproceedings{tang2024synthesis,
  title={Synthesis of data-driven nonlinear state observers using lipschitz-bounded neural networks},
  author={Tang, Wentao},
  booktitle={American Control Conference},
  pages={1713--1719},
  year={2024},
  organization={IEEE}
}

@article{hart2025generic,
  title={Generic and isometric embeddings in reservoir computers},
  author={Hart, Allen G},
  journal={Chaos: An Interdisciplinary Journal of Nonlinear Science},
  volume={35},
  number={11},
  year={2025},
  publisher={AIP Publishing}
}

@inproceedings{2026Saqib,
    author = {Saqib, Danyal and Padoan, Alberto},
    title = {The internal model principle in reservoir computing},
    booktitle = {{arXiv}},
    year = {2026} 
}

@article{marconi2007output,
  title={Output stabilization via nonlinear Luenberger observers},
  author={Marconi, Lorenzo and Praly, Laurent and Isidori, Alberto},
  journal={SIAM Journal on Control and Optimization},
  volume={45},
  number={6},
  pages={2277--2298},
  year={2007},
  publisher={SIAM}
}

@inproceedings{janny2021deep,
  title={Deep {KKL}: Data-driven output prediction for non-linear systems},
  author={Janny, Steeven and Andrieu, Vincent and Nadri, Madiha and Wolf, Christian},
  booktitle={Conference on Decision and Control},
  pages={4376--4381},
  year={2021},
  organization={IEEE}
}

@inproceedings{ramos2020numerical,
  title={Numerical design of Luenberger observers for nonlinear systems},
  author={Ramos, Louise da C and Di Meglio, Florent and Morgenthaler, Valery and da Silva, Lu{\'\i}s F Figueira and Bernard, Pauline},
  booktitle={IEEE Conference on Decision and Control},
  pages={5435--5442},
  year={2020},
  organization={IEEE}
}

@article{isidori2010existence,
  title={About the existence of locally lipschitz output feedback stabilizers for nonlinear systems},
  author={Isidori, Alberto and Praly, Laurent and Marconi, Lorenzo},
  journal={SIAM journal on control and optimization},
  volume={48},
  number={5},
  pages={3389--3402},
  year={2010},
  publisher={SIAM}
}
\end{document}